\documentclass[11pt]{article}   
\usepackage[dvips]{graphicx}
\usepackage{amsmath,amssymb}
\usepackage{algorithm}
\usepackage{algorithmic}

\newtheorem{theorem}{Theorem}
\newtheorem{lemma}{Lemma}
\newtheorem{corol}{Corollary}
\newtheorem{pro}{Proposition}

\date{}

\newenvironment{proof}[1][\hspace{-1.0ex}]%
{\par\addvspace{1mm}{\it Proof\hspace{1.0ex}{#1}.} }%
{\quad$\blacktriangle$\par\addvspace{1mm}}

    \newif\ifNoRemark
    \def\addtheorem#1#2#3#4{ 
    \ifthenelse{\expandafter\isundefined\csname the#2\endcsname}{\newcounter{#2}}{}
    \newenvironment{#1}[1][\global\NoRemarktrue]
     {\par\addvspace{2mm}\noindent 
       \refstepcounter{#2}{\bf #3~\csname the#2\endcsname
      \vphantom{##1}\ifNoRemark.\ \else\ (##1).\fi}\begingroup #4}%
     {\endgroup\par\addvspace{1mm}\global\NoRemarkfalse}
    \expandafter\newcommand\csname b#1\endcsname{\begin{#1}}
    \expandafter\newcommand\csname e#1\endcsname{\end{#1}}
    }

\begin{document}

\title{On  $q$-ary Bent and Plateaued Functions \footnote{\,The work was supported by the program of fundamental scientific
researches of the SB RAS  I.5.1, project N 0314-2019-0017}}

\author{ Vladimir N. Potapov\\
Sobolev Institute of Mathematics
\\ Novosibirsk, Russia \\ vpotapov@math.nsc.ru\\}

\maketitle

\begin{abstract}
 We obtain the following results. For any prime $q$ the minimal Hamming distance  between distinct regular $q$-ary bent
functions of $2n$ variables is equal to $q^n$. The number of $q$-ary
regular bent functions at the distance $q^n$ from the quadratic bent
function $Q_n=x_1x_2+\dots+x_{2n-1}x_{2n}$ is equal to
$q^n(q^{n-1}+1)\cdots(q+1)(q-1)$ for $q>2$.
 The Hamming distance  between distinct binary $s$-plateaued
  functions of $n$ variables is not less than
$2^{\frac{s+n-2}{2}}$ and the Hamming distance between distinct
ternary $s$-plateaued
  functions of $n$ variables is not less than
$3^{\frac{s+n-1}{2}}$.  These bounds are tight.

 For $q=3$ we
prove an upper bound on nonlinearity of ternary functions in terms
of their correlation immunity. Moreover,  functions reaching this
bound are plateaued. For $q=2$ analogous result are well known but
for large $q$ it seems impossible. Constructions and some properties
of $q$-ary plateaued functions are discussed.


\end{abstract}

\section{Introduction}

Boolean bent and plateaued  functions play a significant role in
information theory and combinato\-rics. These functions
 are intensively studied at present as they have numerous applications
 in cryptography, coding theory, and other areas. Bent functions are
 known as Boolean functions with  maximal  nonlinearity.
  $q$-Ary  generalizations
 of bent functions are also interesting mathe\-ma\-tical object (see \cite{Tok2}).
 Boolean plateaued
functions generalize  functions with maximal nonlinearity. Moreover,
some Boolean plateaued functions achieve tradeoff between properties
of nonlinearity and correlation immunity. Recently Mesnager at al.
\cite{Mes} redefined plateaued functions over any finite field $F_q$
where $q$ is a prime power, and established their properties. In
this paper (Section 6) we generalize some methods in order to
construct binary plateaued functions for $q$-ary plateaued
functions.

The Hamming distance $d(f,g)$ between two discrete functions $f$ and
$g$ is the number of arguments where these functions  differ. In
other words, the Hamming distance between two functions $f$ and $g$
is the cardinality of the support $\{x\in Dom(f) \ |\ f(x) \neq g(x)
\}$ of their difference. The problem of  finding  the minimal
Hamming distance between two functions of the same type is known as
the problem of the minimum-support bitrade (see \cite{Kro16}). A
series of results on calculation of the minimal Hamming distance
between Boolean bent and correlation immune functions can be found
in  \cite{Kol2} and \cite{Pot12'}. In this paper (Section 3) we
prove that the minimal Hamming distance between distinct $q$-ary
regular bent functions of $2n$ variables is equal to $q^n$ for any
prime $q$. We show (Section 5) that the number of $q$-ary regular
bent functions at  distance $q^n$ from the quadratic bent function
$Q_n=x_1x_2+\dots+x_{2n-1}x_{2n}$ is equal to
$q^n(q^{n-1}+1)\cdots(q+1)(q-1)$ for all prime $q>2$. In  binary
case the analogous statement was proved in \cite{Kol1}. Moreover we
prove (Section 4) that the Hamming distance between different binary
and ternary $s$-plateaued functions of $n$ variables is at least
$2^{\frac{s+n-2}{2}}$ and $3^{\frac{s+n-1}{2}}$ respectively. We
also construct pairs of  functions at the minimal distance and
investigate their properties\footnote{\,These results were reported
on XV and XVI International Symposia "Problems of Redundancy in
Information and Control Systems".}.

 Nonlinearity and
correlation immunity are well-known cryptographic properties of
discrete functions. In  binary case  these properties are opposite
to each other. Tarannikov \cite{Tar0}  proved the inequality
$nl(f)\leq 2^{n-1}-2^{cor(f)+1}$ for balanced Boolean functions
$f:F^n_2\rightarrow F_2$, where $nl(f)$ is the nonlinearity and
$cor(f)\leq n-2$ is the correlation immunity of $f$. Moreover, he
proved that if the equality holds then $f$ is a plateaued function.
We prove (Section 8) a similar result for ternary functions acting
from $F^n_3$ to $F_3$: $nl(f)\leq 2\cdot3^{n-1}-3^{cor(f)-1}$ and
 $f$ is a ternary plateaued function whenever the equality holds.
Therefore we can use plateaued ternary functions to achieve
tradeoffs between properties of nonlinearity and correlation
immunity. Moreover, we prove (Section 7) an upper bound
$2\cdot3^{n-1}-3^{\frac{n}{2}-1}$ for nonlinearity of arbitrary
ternary functions. The bound is achieved on bent functions of a
special type.

We believe that  the analogous results are impossible for large $q$.
In the last Section  9  we construct quaternary functions having
large both nonlinearity and correlation immunity.

\section{Fourier transform on finite abelian groups}

Let  $G$ be a finite abelian group. Consider  the vector space
$V(G)$ consisting of functions $f:G\rightarrow \mathbb{C}$ with the
inner product
$$(f,g)=\sum\limits_{x\in G}f(x)\overline{g(x)}.$$
A function $f:G\rightarrow \mathbb{C}\backslash\{0\}$ mapping $G$ to
the non-zero complex numbers is called a character of $G$ if it is a
group homomorphism from $G$ to $\mathbb{C}$, i.e.,
$\phi(x+y)=\phi(x)\phi(y)$ for each $x,y\in G$. The set of
characters of an abelian group is an orthogonal basis of $V(G)$. If
$G=(\mathbb{Z}/q\mathbb{Z})^n$  then we can define characters of
$(\mathbb{Z}/q\mathbb{Z})^n$ by equation $\phi_z(x)=\xi^{\langle
x,z\rangle}$, where $\xi=e^{2\pi i/q}$ and $\langle
x,y\rangle=x_1y_1+\dots+x_ny_n\, {\rm mod}\, q$ for each $z\in
(\mathbb{Z}/q\mathbb{Z})^n$.  We may define the Fourier transform of
$f\in V(G)$ by the formula $\widehat{f}(z)=(f,\phi_z)/|G|^{1/2}$,
i.e., $\widehat{f}(z)$ is the coefficients of the expansion of $f$
in the basis of characters. Parseval identity $(f,f)=\|f\|^2=
\|\widehat{f}\|^2$ and the Fourier inversion formula
$\widehat{(\widehat{f(x)})}= f(-x)$ hold. A proof of the following
equation  can be found in  \cite{Tao}.

\begin{pro}[uncertainty principle]\label{sbent1} For every $f\in V(G)$
the following inequality is true:
\begin{equation}\label{ebent1}
 |{\rm
supp}(f)|\cdot|{\rm supp}(\widehat{f})|\geq |G|.
\end{equation}
 \end{pro}

 If $H $ is any subgroup of $G$, and we set $f$ to be the
characteristic function of $H$, then it is easy to see that $|{\rm
supp}(f)|=|H|$ and $|{\rm supp}(\widehat{f})| = |G|/|H|$, so the
bound (\ref{sbent1}) is tight. One can show that up to the
symmetries of the Fourier transform (translation, modulation, and
homogeneity) this is the only way in which (\ref{sbent1}) can be
obeyed with equality.

Define the convolution of  $f\in V(G)$ and $g\in V(G)$  by the
equation
 $f*g(z)=\sum\limits_{x\in G}f(x)g(z-x)$. It is well known that

\begin{equation}\label{ebent2}
\widehat{f*g}=|G|^{1/2} \widehat{f}\cdot \widehat{g}.
\end{equation}

Further, we suppose that
  $q$ is a prime number and
 $G=(\mathbb{Z}/q\mathbb{Z})^n \simeq F^n_q$ is an $n$-dimensional vector space
 over Galois field  $F_q$.

\begin{corol}\label{cbent1}
The equation $|{\rm supp}(f)|\cdot|{\rm supp}(\widehat{f})|=q^n$
holds if and only if $f=c\phi_z\chi[\Gamma]$, where $z\in F^n_q$,
$c\in \mathbb{C}$ is a constant and $\chi[\Gamma]$ is the
characteristic function of an affine space  $\Gamma$ in  $F^n_q$.
\end{corol}

The following equality  can be found in   \cite{Sar} and \cite{Tsf}.
\begin{pro}\label{sbent2}
  If   $\Gamma$ is a linear subspace in
  $F^n_q$ and the subspace $\Gamma^\bot$ is  dual of $\Gamma$
  then it holds
$$\sum\limits_{y\in \Gamma}\widehat{f}(y)=q^{{\rm dim}(\Gamma) -n/2}\sum\limits_{x\in
\Gamma^\bot}{f}(x).$$
\end{pro}

Proposition \ref{sbent2} follows from the equality
$\widehat{f}*\chi[\Gamma]=q^{{\rm
dim}(\Gamma)}\widehat{f\cdot\chi[\Gamma^\bot]}$. It is possible to
obtain a generalization of Proposition \ref{sbent2} by substitution
an argument $a\in F^n_q$ into both parts of the equation.
\begin{equation}\label{eq005}
\sum\limits_{y\in a+\Gamma}\widehat{f}(y)=q^{{\rm dim}(\Gamma)
-n/2}\sum\limits_{x\in \Gamma^\bot}{f}(x)\xi^{-\langle x,a\rangle}.
\end{equation}

It is well known that an extension $\mathbb{Q}(\xi)$ of the field of
rational numbers does not contain roots of unity with  exception of
$\pm \xi^k$ when $q$ is a prime number (see \cite{Lang}). It follows
that

\begin{pro}\label{sbent3} 
1) $\sum\limits_{j=0}^{q-1}\xi^{kj}=0$ if $k\neq 0\, {\rm mod}\,
q$;\\
2)  $\xi$ is not a root of  a rational polynomial  of degree
less than $q-1$;\\
3) if $|\sum_{j=1}^kc_j\xi^{a_j}|\in \mathbb{N}$ and $a_j,c_j\in
\mathbb{Z}$ then $\sum_{j=1}^kc_j\xi^{a_j}=N\xi^b$, where $b,N\in
\mathbb{Z}$.
\end{pro}

\begin{corol}\label{cbent2}If $\sum_{j=1}^kc_j|\xi^{a_j}-\xi^{b_j}|^2=0$, where $a_j,b_j,c_j\in
\mathbb{Z}$, $a_j\neq b_j$, $j=1,\dots,k$, then $\sum_{j=1}^kc_j=0$.
\end{corol}
\begin{proof}
We have
$$|\xi^{a_j}-\xi^{b_j}|^2=(\xi^{a_j- b_j}-1)(\xi^{b_j-a_j}-1)=2-\xi^m-\xi^{q-m},$$
where $m=(a_j-b_j)\mod q$. Therefore,
$$0=\sum_{j=1}^kc_j|\xi^{a_j}-\xi^{b_j}|^2=2\sum_{j=1}^kc_j+
\sum_{m=1}^{q-1}\alpha_m\xi^m,$$ where $\sum_{m=1}^{q-1}\alpha_m=0$.
It follows from Proposition \ref{sbent3} that
$\alpha_m=2\sum_{j=1}^kc_j$ for each $m=1,\dots,q-1$. Then
$\sum_{j=1}^kc_j=0$.
\end{proof}

Complex numbers $\sum_{j=1}^kc_j\xi^{a_j}$ with $a_j,c_j\in
\mathbb{Z}$ are known as Eisenstein  integers if $\xi=e^{2\pi i/3}$
and as Gaussian integers if $\xi=e^{\pi i/2}$. The following
statement is obvious.

\begin{pro}
The absolute value of a nonzero Eisenstein or Gaussian integer is
not less than $1$.
\end{pro}

\section{Bent functions}

A function $f:F^n_q\rightarrow F_q$ is called a $q$-ary bent
function if and only if  $|\widehat{\xi^f}(y)|=1$ for each $y\in
F^n_q$ or equivalently
$\widehat{\xi^f}\cdot\overline{\widehat{\xi^f}}=I$, where $I$ is
equal to
 $1$ everywhere (see \cite{KSW}, \cite{Tok1}). Using
(2) we can obtain that the definition of a bent function is
equivalent to the equation
${\xi^f}*\overline{{\xi^f}}=|G|\chi[\{0\}].$ The matrix
$H=(h_{z,y})$, where $ h_{z,y}=\xi^{f(z-y)}$, is a generalized
Hadamard matrix.

 A bent function  $b$ is called regular if and only if there exists a function
$b':F^n_q\rightarrow F_q$ such that  $\xi^{b'}=\widehat{\xi^{b}}$.
It is easy to see that $b'$ is a bent function as well. In this
section we suppose everywhere that  $n$ is even.

\begin{pro}\label{cbent55} For any pair of $q$-ary regular bent functions $b$
and $b'$, it holds $|{\rm supp}(\xi^b-\xi^{b'})|=|{\rm
supp}(\widehat{\xi^b}-\widehat{\xi^{b'}})|$.
\end{pro}
\begin{proof}
The Fourier transform of  $\xi^b-\xi^{b'}$ is
$\widehat{\xi^b}-\widehat{\xi^{b'}}$. By the Parseval identity  we
obtain that
$\sum_x|\xi^{b(x)}-\xi^{b'(x)}|^2=\sum_y|\widehat{\xi^{b(y)}}-\widehat{\xi^{b'(y)}}|^2$.
It follows from Corollary \ref{cbent2} that in  both sides of the
equation the numbers of nonzero terms are equal.
\end{proof}


\begin{theorem}\label{cbent33}
The Hamming distance between two regular bent functions on  $F^n_q$
is not less than $q^{n/2}$. If it is equal to $q^{n/2}$, then the
difference between these functions is equal to $c\chi[\Gamma]$,
where $c\in F_q$ and $\Gamma$ is an $n/2$-dimensional affine
subspace.
\end{theorem}
\begin{proof}
Let $b,b':F^n_q\rightarrow F_q$ be bent functions. By Proposition
\ref{cbent55}, the equality  $|{\rm supp}(b-b')|= |{\rm
supp}(\widehat{\xi^b}-\widehat{\xi^{b'}})|$ holds.
 It follows from the uncertainty principle
(Proposition \ref{cbent1}) that
$$|{\rm
supp}(b-b')|^2 =   |{\rm supp}(\xi^b-\xi^{b'})||{\rm
supp}(\widehat{\xi^b}-\widehat{\xi^{b'}})|\geq q^n.
$$
\end{proof}

For binary case the statement of Theorem \ref{cbent33}  was proved
in \cite{Car1} and \cite{Kol1}. In \cite{Pot12'} it was found the
spectrum of possible small distances (less than the doubled minimum
distance) between two Boolean bent functions.

\section{Plateaued functions}

Define the Walsh--Hadamard transform of a function
$f:F^n_q\rightarrow F_q$  by the formula 
$W_f=q^{n/2}\widehat{\xi^f}$.

A function $f:F^n_q\rightarrow F_q$ is called a $q$-ary bent
function if and only if  $|W_f(y)|=q^{n/2}$ for each $y\in F^n_q$
and it
 is called a $q$-ary plateaued
function if and only if  $|W_f(y)|\in \{0,\mu\}$ for each $y\in
F^n_q$.

It follows from the Parseval identity  for a plateaued function $f$
that $q^{n}=\sum_x|\xi^{f(x)}|^2=\sum_y|W_f(y)|^2=\mu^2|{\rm
supp}(W_f)| $. Since $q$ is a prime number, we have that
$|W_f(y)|^2$ takes  the value $\mu^2=q^{s+n}$ exactly $q^{n-s}$
times for some $s$. Such $q$-ary plateaued functions $f$ are called
$s$-plateaued. By Proposition \ref{sbent3}(3),   we have $W_f(y)=\pm
q^{(n+s)/2}\xi^a$ for some $a\in F_q$ or $W_f(y)=0$ when
$\frac{n+s}{2}\in \mathbb{Z}$. Note that $0$-plateaued function is a
bent function. We say that an $s$-plateaued function $f$ is regular,
if $W_f(y)=q^{(n+s)/2}\xi^a$ for all $y\in F^n_q$ and
$\frac{n+s}{2}\in \mathbb{Z}$.

  The definition of plateaued functions is equivalent to the
equality ${\xi^f}*\overline{{\xi^f}}*{{\xi^h}}=\mu q^n{{\xi^f}}$,
where $h(x)=f(-x)$ (see \cite{Mes}).

\begin{pro}\label{splat12}
Let $f:F^n_q\rightarrow F_q$ be a $s$-plateaued function,
$A:F^n_q\rightarrow F^n_q$ be a non-degenerate affine transformation
and $\ell:F^n_q\rightarrow F_q$ be an affine function. Then
$g=f\circ A+\ell$ is a $s$-plateaued function.
\end{pro}
\begin{proof}
Suppose that the affine function  $\ell$ is the zero function and
$A(x)=L(x)+u$, where $L$ is a linear transformation and $u\in
F^n_q$. Then it holds $W_g(y)=\sum_x\xi^{g(x)-\langle
x,y\rangle}=\sum_x\xi^{f(L(x)+u)-\langle
x,y\rangle}=\sum_z\xi^{f(z)-\langle L^{-1}(z-u),y\rangle}=
\xi^{\langle L^{-1}u, y\rangle }W_f((L^{-1})^{\rm T}y)$.

Suppose that an affine transform $A$ is identical and
$\ell(x)=\langle x,u\rangle+a$. Then it holds
$W_g(y)=\sum_x\xi^{g(x)-\langle
x,y\rangle}=\sum_x\xi^{f(x)+a-\langle x,y-u\rangle}=\xi^aW_f(y-u)$.
\end{proof}

 Let us calculate the minimal Hamming distance between two
 plateaued functions for binary and ternary cases.

\begin{theorem}\label{cbent3}
1) The Hamming distance between two distinct $s$-plateaued functions
on $F^n_3$
is not less than $3^{\frac{s+n-1}{2}}$.\\
2) The Hamming distance between two distinct $s$-plateaued functions
on  $F^n_2$
 is not less than $2^{\frac{s+n-2}{2}}$.
\end{theorem}
\begin{proof}
1. Let $f,g:F^n_3\rightarrow F_3$ be $s$-plateaued functions. Then
$|\xi^{f(x)}-\xi^{g(x)}|=\sqrt{3}$ if $f(x)\neq g(x)$; and
$|W_f(y)-W_g(y)|=3^{(n+s)/2}|\xi^a\pm \xi^b|$, $a,b\in F_3$ or
$|W_f(y)-W_g(y)|=3^{(n+s)/2}|\xi^a-0|$ if $W_f(y)\neq W_g(y)$. In
both cases it holds $|W_f(y)-W_g(y)| \geq 3^{(n+s)/2}$. Then, by the
Parseval identity,  we obtain that
$$3|{\rm supp}(f-g)|=\|\xi^{f}-\xi^{g}\|^2=\frac{1}{3^n}\|W_f-W_g\|^2\geq 3^{s}|{\rm
supp}(W_f-W_g)|.$$ It follows from the uncertainty principle
(Proposition \ref{sbent1}) that
$$3^{1-s}|{\rm
supp}(f-g)|^2 \geq   |{\rm supp}(f-g)||{\rm supp}(W_f-W_g)|\geq 3^n.
$$

2. Let $f,g:F^n_2\rightarrow F_2$ be $s$-plateaued functions. Then
$|\xi^{f(x)}-\xi^{g(x)}|=2$ if $f(x)\neq g(x)$; and\\
$|W_f(y)-W_g(y)|=2^{(n+s)/2}$ or $|W_f(y)-W_g(y)|=2^{(n+s+2)/2}$ if
$W_f(y)\neq W_g(y)$. Acting similar to  the case $q=3$ we obtain the
inequality $2^{2-s}|{\rm supp}(f-g)|^2 \geq  2^n.$
\end{proof}

If the  distance between $s$-plateaued functions is equal to
$2^{\frac{s+n-2}{2}}$ ($3^{\frac{s+n-1}{2}}$) then the difference
between these functions is equal to $c\chi[\Gamma]$, where $c\in F_2
(F_3)$ and $\Gamma$ is an $\frac{s+n-2}{2}$
($\frac{s+n-1}{2}$)-dimensional affine subspace (see Proposition
\ref{sbent1}). Moreover, in this case both of these $s$-plateaued
functions are affine functions on $\Gamma$.

From Proposition 2 one can conclude that the following statements
hold.

\begin{pro}\label{cbent41}
1) If an $s$-plateaued function $f:F^n_q\rightarrow F_q$ coincide
with an affine function on
 an affine subspace
 $\Gamma$, then ${\rm dim}\Gamma \leq \frac{s+n}{2}$.

2) If an $s$-plateaued function $f:F^n_q\rightarrow F_q$ coincide
with an affine function on
 an $\frac{s+n}{2}$-dimensional affine subspace, then there exist
  $q-1$ $s$-plateaued functions that differ from
  $f$ only on this subspace.
\end{pro}
\begin{proof}
By Proposition \ref{splat12}, without lost of generality, we suppose
that $\Gamma$ is a
 linear subspace and $f$ is equal to $0$
everywhere on $\Gamma$, i.\,e., $f|_\Gamma=0$. By Proposition
\ref{sbent2}, we obtain that $\sum_{y\in\Gamma^\perp}W_f(y)=q^n$. By
definition of $s$-plateaued functions, we get that
$|W_f(y)|=q^{(n+s)/2}$. Thus $|\Gamma^\perp|\geq q^{(n-s)/2}$. Then
${\rm dim}\Gamma \leq n-\frac{n-s}{2}=\frac{s+n}{2}$.

If ${\rm dim}\Gamma = \frac{s+n}{2}$ then  from the equation
$\sum_{y\in\Gamma^\perp}W_f(y)=q^n$ we obtain that
$W_f|_{\Gamma^\perp}=q^{(n+s)/2}$.

Consider function $g=a{\chi[\Gamma]}$, $a\in F_q$.
 It is easy to see that $\xi^f\xi^g=\xi^f+(\xi^a-1)\chi[\Gamma]$
 and
$(\xi^a-1)\widehat{\chi[\Gamma]}=q^{s/2}(\xi^a-1)\chi[\Gamma^\perp]$.
Next,
$$W_{f+g}=q^{n/2}\widehat{\xi^f\xi^g}=q^{n/2}\widehat{\xi^f}+q^{n/2}(\xi^a-1)\widehat{\chi[\Gamma]}=
W_f+q^{(n+s)/2}(\xi^a-1)\chi[\Gamma^\perp].$$ Consequently,
$|W_{f+g}(y)|=q^{(n+s)/2}$ for each $y\in F^n_q$.
\end{proof}
Note that if  a regular $s$-plateaued function $f$ coincide with
some affine function on $\Gamma$ then $f+g$, where $g$ is defined in
the above proof, is also regular.

Proposition  \ref{cbent41} was proved in \cite{Car1} for Boolean
bent functions.

\section{Quadratic forms}

A quadratic form $Q:F_q^m\rightarrow F_q$ is called non-degenerate
if and only if   $\{x\in F_q^m : \forall y\in F_q^m,
Q(y+x)=Q(y)\}=\{\overline{0}\}$. A linear subspace $U$ in $F_q^m$ is
called totally isotropic  for form $Q$ if and only if $Q(U)=0$. The
largest dimension of a totally isotropic subspace is often called
the Witt index of the form.
 If
$m=2n$, then the maximal Witt index of  non-degenerate forms of
degree $m$ is equal to $n$. All non-degenerate forms with the
maximal Witt index are equivalent with respect to non-degenerate
linear transformations of arguments. One of such quadratic forms
$Q_n$ is determined by the equation $Q_n(v_1,\dots,v_n,u_1,\dots
u_n)= v_1u_1+\dots+v_nu_n$. It is well known that $Q_n$ is a regular
bent function from Maiorana--McFarland class (see \cite{MF},
\cite{Mes0} and \cite{Tok1}). The following proposition is proved,
for example, in \cite{BCN} (p.274, Lemma 9.4.1).

\begin{pro}\label{sbent4}
The number of  totally isotropic subspaces of $Q_n$  is equal to
$\prod\limits_{i=1}^{n}(q^{n-i}+1)$.

\end{pro}

It is easy to see that if $Q_n$ is an affine function  on  some
affine subspace, then it is an affine function on every coset.
Moreover, if $Q_n$ is an affine function  on a linear subspace of
 dimension $n$, then this subspace is isotropic
(here we assume that $q>2$). Thus $Q_n$ is an affine function on all
cosets of totally isotropic subspaces and  it is not affine  on
other linear subspaces of dimension $n$.

From Theorem 1, Propositions \ref{cbent41}(2) and \ref{sbent4} we
can conclude that the following holds.

\begin{corol}\label{cbent5}
If $q$ is a prime number and $q>2$, then there are exactly
$q^n(q^{n-1}+1)\cdots(q+1)(q-1)$ $q$-ary regular bent functions  at
distance $q^{n}$ from  $Q_n$.
\end{corol}

For binary case an analogous statement was proved in \cite{Kol1}. In
\cite{Kol2} it was established that this bound  on the number of
Boolean bent functions at the minimal distance is achieved only for
quadratic bent functions. It is natural to propose that this
property of quadratic bent functions holds for any prime $ q\geq2$.

\section{Constructions of plateaued functions}

 Denote by
$x^k=(x_1,x_2,\dots,x_k)$ the left part of a vector
$x=(x_1,x_2,\dots,x_k,\dots,x_n)$. The following construction of
plateaued functions is similar to Maiorana--McFarland's construction
for bent functions.

\begin{pro}\label{cplat133}
Let $x,y,u,v\in F^n_q$, $z,w\in F^k_q$, $k\leq n$. Let
$\tau:F^n_q\rightarrow F^n_q$ and  $\sigma:F^k_q\rightarrow F^k_q$
be bijections and $f:F^n_q\rightarrow F_q$ be an arbitrary function.
Then the function $F(x,y,z)=\langle\tau(x),y\rangle +
\langle\sigma(x^k),z\rangle +f(x)$ is a $k$-plateaued function of
$2n+k$ variables.
\end{pro}
\begin{proof}
$W_F(u,v,w)= \sum\limits_{x,y,z\in
F^n_q}\xi^{\langle\tau(x),y\rangle + \langle\sigma(x^k),z\rangle
+f(x)-\langle x,u\rangle-\langle v,y\rangle -\langle w,z\rangle}=$\\
$\sum\limits_{x\in F^n_q}\xi^{f(x)-\langle
x,u\rangle}\sum\limits_{y\in
F^n_q}\xi^{\langle\tau(x),y\rangle-\langle
v,y\rangle}\sum\limits_{z\in F^k_q}\xi^{\langle\sigma(x^k),z\rangle
-\langle w,z\rangle}$.

By Proposition \ref{sbent3}(1) we obtain that the sum
$\sum\limits_{z\in F^k_q}\xi^{\langle\sigma(x^k),z\rangle -\langle
w,z\rangle}$  is equal to $q^k$ if $w=\sigma(x^k)$ and  equal to $0$
otherwise. In the same way,
 the sum
$\sum\limits_{y\in F^n_q}\xi^{\langle\tau(x),y\rangle -\langle
v,y\rangle}$  is equal to $q^n$ if $v=\tau(x)$ and equal to $0$
otherwise.

Therefore  $W_F(u,v,w)=q^{n+k}\xi^{f(x)-\langle x,u\rangle}$ if
$v=\tau(x)$ and $w=\sigma(x)$;  $W_F(u,v,w)=0$ otherwise.
\end{proof}

Similarly, it is possible to construct many $((t-2)n+k)$-plateaued
functions of $tn+k$ variables.

\begin{corol}\label{cplat01}
The number of different $q$-ary $((t-2)n+k)$-plateaued functions of
$tn+k$ variables is not less than $q^{q^n}(q^n!)^{t-1}q^k!$.
\end{corol}

Define the function $\theta:F_q^k\rightarrow \{0,1\}$  by the
equation $\theta=\chi[\{\overline{0}\}] $, i.\,e., $\theta(y)=1$ if
$y=\overline{0}$ and $\theta(y)=\overline{0}$ if $y\neq
\overline{0}$.

\begin{pro}\label{cplat13}

1) Let $g$ be a  function of $n+k$ variables and
$f(\overline{x},y)=g(\overline{x},a)$, $a\in F^k_q$. Then
$W_f(\overline{u},v)=q^kW_g(\overline{u},a)$ if $v=\overline{0}$ and
$W_f(\overline{u},v)=0$ if $v\neq\overline{0}$.

2) Suppose $g(\overline{x},y)=f(\overline{x})+\langle a, y\rangle$,
$a\in F^k_q$. Then
$W_g(\overline{u},v)=q^k\theta(a-v)W_f(\overline{u})$.

3) Suppose $g(\overline{x},y)=\sum_a\theta(a+y)f^a(\overline{x})$,
$a\in F^k_q$. Then $W_g(\overline{u},v)=\sum_a\xi^{\langle
a,v\rangle}W_{f^a}(\overline{u})$.
\end{pro}
\begin{proof}
1. $W_f(\overline{u},v)=\sum\limits_{\overline{x}\in F^n_q, y\in
F^k_q}\xi^{g(\overline{x},a)}\xi^{-\langle
v,y\rangle-\langle\overline{u},\overline{x}\rangle}=
\sum\limits_{\overline{x}\in
F^n_q}\xi^{g(\overline{x},a)}\xi^{-\langle\overline{u},\overline{x}\rangle}\sum\limits_{
y\in F^k_q}\xi^{-\langle v,y\rangle}$. By Proposition
\ref{sbent3}(1), we obtain that $\sum\limits_{ y\in
F^k_q}\xi^{-\langle v,y\rangle}=q^k$ if $v=\overline{0}$ and it is
equal to $0$ otherwise.

2. $W_g(\overline{u},v)=\sum\limits_{\overline{x}\in F^n_q, y\in
F^k_q}\xi^{g(\overline{x},y)}\xi^{-\langle
v,y\rangle-\langle\overline{u},\overline{x}\rangle}=
\sum\limits_{\overline{x}\in
F^n_q}\xi^{f(\overline{x})-\langle\overline{u},\overline{x}\rangle}\sum\limits_{y\in
F^k_q}\xi^{\langle a,y\rangle-\langle v,y\rangle}$. By Proposition
\ref{sbent3}(1), we obtain that $\sum\limits_{y\in
F^k_q}\xi^{\langle a, y\rangle-\langle v,y\rangle}=q^k\theta(a-v)$.

3. $W_g(\overline{u},v)= \sum\limits_{a\in
F^k_q}\sum\limits_{\overline{x}\in
F^n_q}\xi^{f^a(\overline{x})+\langle
v,a\rangle-\langle\overline{u},\overline{x}\rangle}=\sum\limits_{a\in
F^k_q}\xi^{\langle v,a\rangle}W_{f^a}(\overline{u}).$
\end{proof}

\begin{corol}\label{cplat1}
1) If $f$ is an $s$-plateaued function of $n+k$ variables and $a\in
F^k_q$ then $g(\overline{x},y)=f(\overline{x},a)$ is an
$(s+k)$-plateaued function of $n+k$ variables.

2) If $f$ is an $s$-plateaued function of $n$ variables and $a\in
F^k_q$ then $f(\overline{x})+\langle a,y\rangle$ is an
$(s+k)$-plateaued function of $n+k$ variables.

3) If $f^a$ are $s$-plateaued functions of $n$ variables with
pairwise disjoint supports of $W_{f^a}$, $a\in F^k_q$, then
$\sum_a\theta(a+y)f^a(\overline{x})$ is an $(s-k)$-plateaued
function of $n+k$ variables.
\end{corol}

For $q=2$ these constructions of plateaued functions and the example
below can be found in \cite{Tar2} and \cite{Tar1}.

Functions $R(x_1,x_2,x_3,x_4)=(x_1+x_2)(x_3+x_4)+x_1+x_3$ and
$R'(x_1,x_2,x_3,x_4)=(x_1+x_4)(x_3+x_2)+x_1+x_3$ are $2$-plateaued
functions over $F_2$ at distance $4=2^{\frac{s+n-2}{2}}$. Functions
$T(x)=x^2$ and $T'(x)=x+2x^2$ are $0$-plateaued functions over $F_3$
at distance $1=3^{\frac{s+n-1}{2}}$.

Using functions $R,R',T,T'$  and Corollary \ref{cplat1}, we prove
the following statement.

\begin{pro}
1) For any integer $s\geq 2$ and $t\geq 0$, there exist pairs of
$2$-ary
 $s$-plateaued functions of $n=2+s+2t$ variables at  distance
$2^{\frac{s+n-2}{2}}$.

2) For any integer $s\geq 0$ and $t\geq 0$, there exist pairs of
$3$-ary $s$-plateaued functions of $n=1+s+2t$ variables at distance
$3^{\frac{s+n-1}{2}}$.
\end{pro}
\begin{proof}
1. For $a\in F^{k_1}_2$ consider $(2+k_1)$-plateaued functions
$g_a({x},y)=R(x)+\langle a,y\rangle$ and $g'_a({x},y)=R'(x)+\langle
a,y\rangle$ of $4+k_1$ variables. By the definition,
$d(g_a,g'_a)=4\cdot2^{k_1}=2^\frac{4+k_1+2+k_1-2}{2}$. It is easy to
see that $W_{g_a}(u,v)=0$ if $v\neq a$. Then functions $W_{g_a}$
have pairwise disjoint supports and functions $W_{g_a}$ and
$W_{g'_c}$ have disjoint support if $a\neq c$.
 Consider the set
of functions $G_{k_2}=\{g_a : a=(b_1,\dots,b_{k_2},0,\dots,0)\}$.
Using the set $G_{k_2}$ and Proposition \ref{cplat1}(2),   it is
possible to construct a $(2+k_1-k_2)$-plateaued function $h$ of
$(4+k_1+k_2)$ variables. Consider the set
$G'_{k_2}=G_{k_2}\cup\{g'_0\}\setminus \{g_0\}$. Acting the same
way, we can construct a $(2+k_1-k_2)$-plateaued function $h'$  using
$G'_{k_2}$. It is easy to see that
$d(h,h')=4\cdot2^{k_1}=2^\frac{4+k_1+k_2+2+k_1-k_2-2}{2}$. Thus for
all integers $s\geq 2$ and $t\geq 0$ we have found pairs of
$s$-plateaued functions  of $n=2+s+2t$ variables at distance
$2^{\frac{s+n-2}{2}}$, where $k_1=s-2+t$ and $k_2=t$.

2. The proof of ternary case is similar to the proof of binary case.
\end{proof}

\section{Nonlinearity of ternary functions}

Denote by $A_{n,q}$ the set of affine functions  $f:F^n_q\rightarrow
F^n_q$. Functions $f,g:F^n_q\rightarrow F^n_q$ are said to be
isotopic if and only if there exists a set of permutations
$\tau_i:F_q\rightarrow F_q$, $i=0,\dots,n$, such that
$g(x_1,\dots,x_n)=\tau_0g(\tau_1x_1,\dots,\tau_nx_n)$ for all
$(x_1,\dots,x_n)\in F_q^n$. Denote by $\widetilde{A}_{n,q}$ the set
of functions that are isotopic to affine functions.

The distance between a function $f$ and a set of functions $A$ is
the minimum distance between $f$ and any function $g\in A$. The
nonlinearity $nl(f)$ of a function $f:F^n_q\rightarrow F^n_q$ is the
distance between $f$ and $A_{n,q}$ and the strong nonlinearity
$\widetilde{nl}(f)$ is the distance between $f$ and
$\widetilde{A}_{n,q}$.

For $q=2$ and $3$ the sets  $A_{n,q}$ and $\widetilde{A}_{n,q}$
coincide. Thus the strong nonlinearity is the same as the
nonlinearity in these cases.

The nonlinearity of a function $f$ mapping  from $F_2^n$ to $F_2$ is
expressed via its Walsh--Hadamard coefficients by the following
formula (see \cite{Tok1} or \cite{Tar1})
\begin{equation}\label{eplat2}
nl(f)=2^{n-1}-2^{-1}\max\limits_{u\in F_2^n}|W_f(u)|.
\end{equation}

The similar formula exists for ternary functions $f:F^n_3\rightarrow
F^n_3$.
\begin{pro}\label{proplat11}  $nl(f)=2(3^{n-1}-3^{-1}\max\limits_{a\in F_3, y\in F_3^n}Re(\xi^aW_f(y))).$
\end{pro}
\begin{proof}
By definitions, it holds $W_f(y)=\sum\limits_x\xi^{f(x)-\langle
x,y\rangle}$. If $f(x)\neq \langle x,y\rangle$ then
$\xi^{f(x)-\langle x,y\rangle}=\xi$ or $\xi^{f(x)-\langle
x,y\rangle}=\xi^2$. In  both cases $Re(\xi^{f(x)-\langle
x,y\rangle})=-\frac12$. If $f(x)= \langle x,y\rangle$ then
$\xi^{f(x)-\langle x,y\rangle}=1$. Consequently, $Re(W_f(y))=|\{x\in
F^n_3:f(x)= \langle x,y\rangle\}|-\frac12|\{x\in F^n_3:f(x)\neq
\langle x,y\rangle\}|$. By the obvious equality $|\{x\in F^n_3:f(x)=
\langle x,y\rangle\}|+|\{x\in F^n_3:f(x)\neq \langle
x,y\rangle\}|=3^n$, we obtain that $d(f,\langle
x,y\rangle)=2(3^{n-1}-3^{-1}Re(W_f(y)))$.

Consider an affine function $\langle x,y\rangle-a$, where $a=1,2$.
We have  that $d(f,\langle x,y\rangle-a)=d(f+a,\langle x,y\rangle)$.
Then $d(f+a,\langle x,y\rangle)=2(3^{n-1}-3^{-1}Re(\xi^aW_f(y))$ as
is proven above.
\end{proof}

A Boolean function $f:F^n_2\rightarrow F_2$ reachs maximum
nonlinearity $nl(f)=2^{n-1}-2^{\frac{n}{2}-1}$ if and only if $n$ is
even and $f$ is a bent function. For $q$-ary bent functions the
similar statement is not true. Consider the case $q=3$ in details.

\begin{pro}\label{proplat13}
1) For any $f:F^n_3\rightarrow F_3$ it holds $nl(f)\leq
2\cdot3^{n-1}-3^{\frac{n}{2}-1};$\\
2) $nl(f)= 2\cdot3^{n-1}-3^{\frac{n}{2}-1}$ if and only if $f$ is a
bent function, $n$ is even and $W_f(y)=-3^{n/2}\xi^a$, $a\in F_3$,
for
each $y\in F^n_3$;\\
 3) if $b$ is a bent function and there exists $y\in F_3^n$ such that
 $W_b(y)=3^{n/2}\xi^a$, $a\in F_3$, then $nl(b)=2(3^{n-1}-3^{\frac{n}{2}-1}).$
\end{pro}
\begin{proof}
 By the Parseval identity, there exists $y\in F_3^n$ such that
$|W_f(y)|=3^{n/2}$. It is easy to see that $\max\limits_{a\in
F_3}Re(\xi^aN)\geq |N|/2$ for all $N\in \mathbb{C}$ and
$\max\limits_{a\in F_3}Re(\xi^aN)=|N|/2$ if and only if
$N=-|N|\xi^a$, $a\in F_3$. A nonlinearity is an integer number.
Consequently, if $nl(f)= 2\cdot3^{n-1}-3^{\frac{n}{2}-1}$ then $n$
is even. Consequently, items 1) and 2) is true.

If $b$ is a bent function then $\max\limits_{a\in
F_3}Re(\xi^aW_b(y))\leq |W_b(y)|=3^{n/2}$ for each $y\in F^n_3$. If
$W_b(y)=3^{n/2}\xi^a$ for some $y\in F^n_3$ then $\max\limits_{a\in
F_3}Re(\xi^aW_f(y))=3^{n/2}$.
\end{proof}

The function $b(x)=x_1^2+\dots +x_n^2$ satisfies the conditions of
Proposition \ref{proplat13} (2) for $n\equiv2(\mod 4)$. Indeed it
holds $W_b(y)=\sum\limits_{x\in F^n_3}\xi^{b(x)-\langle
y,x\rangle}=\prod\limits_{j=1}^{n}\sum\limits_{x_j\in
F_3}\xi^{x_j(x_j-y_j)}$. By direct calculations, we obtain that
$\sum\limits_{a\in F_3}\xi^{a(a-u)}=(1-\xi)\xi^{1-u^2}$ and
$(1-\xi)^2=-3\xi$. Thus for all $y\in F^n_3$ we have
$W_b(y)=(-3)^{n/2}\xi^{a(y)}$ for some $a(y)\in F_3$.

\section{Extremal property of plateaued functions}

Consider a function $f:F_q^n\rightarrow F_q$. If for all $a\in F_q$
numbers $|f^{-1}(a)\cap \Gamma|$ are equal for all $k$-dimensional
faces $\Gamma$ then $f$ is called a correlation immune function of
order $n-k$. Denote by $cor(f)$ the maximum of these orders $n-k$.
The correlation immunity of a balanced ($W_f(\bar 0)=0$) function
$f$ on $F_q^n$ is expressed via its Walsh--Hadamard coefficients by
the following formula (see \cite{Pot} and \cite{Tar0})
\begin{equation}\label{eplat2}
cor(f)=\min\limits_{u\in F_q^n, W_f(u)\neq 0}wt(u)-1,
\end{equation}
where $wt(u)$ is the Hamming weight of  $u$.

The following theorem is proved by Tarannikov in \cite{Tar0}.

\begin{theorem}\label{tplat3} Let $f$ be a balanced Boolean function on
$F_2^n$, $cor(f)\leq n-2$. Then $nl(f)\leq 2^{n-1}-2^{cor(f)+1}$. If
$nl(f)= 2^{n-1}-2^{cor(f)+1}$  then $f$ is a plateaued function.
\end{theorem}

Therefore we can use plateaued functions to achieve tradeoff between
 nonlinearity and correlation immunity
 properties of Boolean functions.
 Constructions described in Proposition \ref{cplat13} save the property
 to be a binary plateaued function with maximal correlation immunity (see
 \cite{Tar1}).

Further, we  generalize Theorem \ref{tplat3} to the ternary
alphabet. Denote by $S^k=\{x\in F^k_3 : wt(x)=k\}$ the set of
ternary vectors of the maximal weight.

\begin{lemma}
For each $x\in S^k$  there exist a set of affine hyperplanes
$\{C_i\}$ of $F^k_3$ and a set of coefficients $\delta_i=\pm 1$ such
that a linear combination over $\mathbb{R}$ of the type
$\sum_i\delta_i\chi[C_i\cap S^k]$  is equal to $2^{k-1}$ on $x$ and
$0$ otherwise.
\end{lemma}
\begin{proof}
It is sufficient to prove the lemma for $x=\bar 1^k=(1,\dots,1)$
because we can replace each $x\in S^k$ by $\bar 1^k$ multiplying
variables by $2$. For $k=1$ the statement of lemma is obvious.
Assume that it holds for $k=n$, i.\,e., $2^{k-1}\chi[\{\bar
1^k\}\cap S^k]= \sum_i\delta_i\chi[C_i\cap S^k]$, where
$\delta_i=\pm 1$.

For an arbitrary affine hyperplane $M\subset F^n_3$ we determine
three affine hyperplanes in $F^{n+1}_3$:\\
 $R_1(M)=\{(x_1,\dots,x_n,y) : y\in F_3, (x_1,\dots,x_n)\in M\}$;\\
 $R_2(M)=\{(x_1,\dots,x_{n-1},y,x_n+1) : y\in F_3, (x_1,\dots,x_n)\in M\}$;\\
 $R_3(M)=\{(x_1,\dots,x_{n-1},x_n+y+2,y) : y\in F_3, (x_1,\dots,x_n)\in M\}$.

By induction assumption,  there exists a set affine hyperplanes
$C_i\subset F^n_3$ such that\\ $2^{k-1}\chi[\{\bar 1^k\}\cap S^k]=
\sum_i\delta_i\chi[C_i\cap S^k]$, $\delta_i=\pm 1$. Then we obtain
equations:\\
 $2^{k-1}\chi[\{(\bar 1^{k-1},1,y):y=1,2\}]=
\sum_i\delta_i\chi[R_1(C_i)\cap S^{k+1}]$;\\
$2^{k-1}\chi[\{(\bar 1^{k-1},y,2):y=1,2\}]=
\sum_i\delta_i\chi[R_2(C_i)\cap S^{k+1}]$;\\
$2^{k-1}\chi[\{(\bar 1^{k-1},y,y):y=1,2\}]=
\sum_i\delta_i\chi[R_3(C_i)\cap S^{k+1}]$. \\
It holds $2\chi[\{(\bar 1^{k-1},1,1)\}]=$\\
$=\chi[\{(\bar 1^{k-1},1,y):y=1,2\}]-\chi[\{(\bar
1^{k-1},y,2):y=1,2\}]+\chi[\{(\bar 1^{k-1},y,y):y=1,2\}]$. The step
of induction is proved.
\end{proof}

\begin{corol}\label{corplat1}
For each $x\in S^k$  there exists a set of affine hyperplanes
$\{C_i\}$ of $F^k_3$ such that
$2^{k-1}\chi[\{x\}]=\sum\limits_i\delta_i\chi[C_i]
+\sum\limits_{y\in F^k_3\backslash S^k}\varepsilon(y)\chi[\{y\}]$,
where $\delta_i$ and $\varepsilon(y)$ are integers.
\end{corol}

Since $1+\xi+\xi^2=0$, a representation of an Eisenstein integer in
the form $a_0+a_1\xi+a_2\xi^2$ is not unique. Moreover, three
irreducible representations of an Eisenstein integer in forms
$N=b_{20}+b_{21}\xi=b_{10}+b_{12}\xi^2=b_{01}\xi+b_{02}\xi^2$ are
unique. Consider an Eisenstein integer $N=p(a_0+a_1\xi+a_2\xi^2)$,
where $p,a_0,a_1,a_2$ are integers. Then
$$N=p(a_0-a_2)+p(a_1-a_2)\xi=p(a_0-a_1)+p(a_2-a_1)\xi^2=p(a_1-a_0)\xi+p(a_2-a_0)\xi^2.$$
Consequently, if there exists a representation of Eisenstein integer
such that its coefficients are divisible by $p$ then coefficients of
all irreducible  representations are divisible by $p$. An Eisenstein
integer $N$ of such type is called divisible by $p$.

\begin{lemma}\label{lemmaplat2}
Let $f$ be a balanced ternary function on $F_3^n$. Then each
coefficient $W_f(z)$ is divisible by $3^{cor(f)}$.
\end{lemma}
\begin{proof}
From (\ref{eplat2}) it follows that  $W_f(z)=0$ if $wt(z)\leq
cor(f)=k$. From (\ref{eq005}) it follows that $ \sum\limits_{y\in
a+\Gamma}W_f(y)=3^{k}\sum\limits_{x\in
\Gamma^\bot}{f}(x)\xi^{-\langle x,a\rangle}$ for each $a\in F^n_3$
and for each $k$-dimensional subspace $\Gamma$. Consider $z\in
F^n_3$ such that $wt(z)=k+1$. By Corollary \ref{corplat1}, we obtain
that $2^{k}W_f(z)=\sum\limits_i\sum_{y\in C_i}\delta_iW_f(y)
+\sum\limits_{wt(y)<k}\varepsilon(y)W_f(y)$. Therefore,
$2^{k}W_f(z)=3^kN$, where $N$ is an Eisenstein integer. Hence the
Eisenstein integer $W_f(z)$ is divisible by $3^k$.

Let us prove  by induction that Eisenstein integers $W_f(x)$ are
divisible by $3^k$ for $wt(x)=m$, $m=k+1,\dots,n$. By Corollary
\ref{corplat1} we obtain that $2^{k}W_f(x)=\sum\limits_i\sum_{y\in
C_i}\delta_iW_f(y) +\sum\limits_{wt(y)<m}\varepsilon(y)W_f(y)$,
where $C_i$ are $k$-dimensional affine subspaces.  From
(\ref{eq005}) it follows that the sum $\sum\limits_{y\in
C_i}\delta_iW_f(y)$ is divisible by $3^k$. If $wt(y)<m$ then
coefficients $W_f(y)$ are divisible by $3^k$ by the induction
assumption.
\end{proof}

\begin{theorem} Let $f$ be a balanced ternary function on
$F_3^n$. Then $nl(f)\leq 2\cdot3^{n-1}-3^{cor(f)-1}$. If $nl(f)=
2\cdot3^{n-1}-3^{cor(f)-1}$ then $f$ is a plateaued function.
\end{theorem}
\begin{proof}
By Lemma \ref{lemmaplat2}, it holds
$\xi^bW_f(z)=3^{cor(f)}(a_0\xi^b+a_1\xi^{1+b}+a_2\xi^{2+b})$ for
each $z\in F^n_3$. Then we get that
$Re(W_f(z))=3^{cor(f)}\frac12(2a_0-a_1-a_2)$, $ Re(\xi
W_f(z))=3^{cor(f)}\frac12(2a_2-a_1-a_0)$ and
$Re(\xi^2W_f(z))=3^{cor(f)}\frac12(2a_1-a_0-a_2)$. If $W_f(z)\neq 0$
then $\max_{b\in F_3}(Re(\xi^bW_f(z)))\geq 3^{cor(f)}/2$.  By
Proposition \ref{proplat11}, we obtain that $nl(f)\leq
2\cdot3^{n-1}-3^{cor(f)-1}$.

It is clear that if $\max_{b\in F_3}(Re(\xi^bN))= M/2$ and an
Eisenstain integer $N$ is divisible by integer $M$ then $N=-M\xi^b$,
where $b\in F_3$. Hence if $\max_{b\in F_3}(Re(\xi^bW_f(z)))=
3^{cor(f)}/2$ then $W_f(z)=-3^{cor(f)}\xi^b$, where $b\in F_3$.
Consequently, $f$  is a plateaued function by the definition.
\end{proof}

\section{Nonlinearity of $q$-ary functions as $q\geq 4$}

In this section we prove that  for $q\geq 4$  there exist $q$-ary
functions $f$ of $n$ variables with
 $cor(f)=n-1$ and with a large nonlinearity.
 By definitions, arbitrary  $n$-ary quasigroup $f$ of order $q$ is
  a $q$-ary function of $n$ variables with
 $cor(f)=n-1$  and each balanced $q$-ary function of $n$ variables with
 $cor(f)=n-1$ is an $n$-ary quasigroup $f$ of order $q$. It is well
 known that all $n$-ary quasigroup $f$ of order $3$ are isotopic to
a linear $n$-ary quasigroup (see, for example, \cite{PK}). For
$q\geq 4$ there exist nonlinear  quasigroups (latin squares). Define
an $n$-ary quasigroup $H$ of order $q$ by the equality
$H(x_1,\dots,x_n)=h(x_1,h(x_2,\dots, h(x_{n-1},x_n)\dots)$. It is
easy to see that $d(H,\widetilde{A}_{n,q})\geq
d(h,\widetilde{A}_{2,q})q^{n-2}$.  If $h$ is not isotopic to a
linear quasigroup then $d(h,\widetilde{A}_{2,q})\geq 4$ because
$d(h,f)\geq 4$ for any two different quasigroups $h$ and $f$.
Further, we prove more tight bound on the strong nonlinearity in the
case $q=4$.

 Consider elements of $F_4$ as pairs of elements of $F_2$.
Each semilinear $n$-ary quasigroup of order
 $4$ can be define (up to isotopy)  by the following formula (see \cite{PK})
\begin{equation}\label{eqplat33}
  f((x_1,y_1),\dots,(x_n,y_n))=(x_1\oplus\dots\oplus
 x_n,b(y_1,\dots,y_n)),
 \end{equation}
 where $b:F_2^n\rightarrow F_2$ is an appropriate Boolean function.
\begin{pro}
If a function $f$ is defined by (\ref{eqplat33}) with the help of a
Boolean bent function $b$ ($n$ is even) then
$\widetilde{nl}(f)=2^n(2^{n-1}-2^{\frac{n}{2}-1})$.
\end{pro}
\begin{proof}
It is easy to see that if an affine function does not depend on one
or more coordinates then the distance between an $n$-ary quasigroup
of order $4$ and such an affine function is equal to $3\cdot
4^{n-1}$.

If an affine function $\ell:F^n_4\rightarrow F_4$ depends on all
coordinates then it is isotopic to the function
$\ell_0((x_1,y_1),\dots,(x_n,y_n))=(x_1\oplus\dots\oplus
 x_n,y_1\oplus\dots\oplus y_n)$.

Consider two cases 1) $\ell((0,y_1),\dots,(0,y_n))\in
\{(0,1),(0,0)\}$ for all $(y_1,\dots,y_n)\in F^n_2$, 2) otherwise.

For the first case we obtain that
$\ell((x_1,y_1),\dots,(x_n,y_n))=(x_1\oplus\dots\oplus
 x_n,l(y_1,\dots,y_n))$, where $l:F^n_2\rightarrow F_2$ is an affine
 Boolean function (see \cite{PK}). Then $d(f,\ell)=2^nd(b,l)\geq 2^n(2^{n-1}-2^{\frac{n}{2}-1})$.

In the second case we get that $|\{(y_1,\dots,y_n)
:\ell((\delta_1,y_1),\dots,(\delta_n,y_n))\in \{(0,1),(0,0)\}\}|=$
  $ |\{(y_1,\dots,y_n)
:\ell((\delta_1,y_1),\dots,(\delta_n,y_n))\not \in
\{(1,1),(1,0)\}\}|=2^{n-1}$ for all  $(\delta_1,\dots,
 \delta_n)\in F^n_2$ (see \cite{PK}). Then
 $d(f,\ell)\geq2^n\cdot2^{n-1}$.
\end{proof}

\end{document}